\documentclass[letterpaper,11pt]{article}
\usepackage[utf8]{inputenc}

\usepackage{style}
\usepackage{dsfont}

\usepackage[margin=1in,footskip=0.25in]{geometry}

\title{Sharp exact recovery threshold for two-community Euclidean random graphs }
\author{Julia Gaudio\thanks{(\url{julia.gaudio@northwestern.edu}) Department of Industrial Engineering and Management Sciences, Northwestern
University} \and Charlie K. Guan\thanks{(\url{charlie.guan@northwestern.edu}) Department of Industrial Engineering and Management Sciences, Northwestern
University}
}

\date{}
\begin{document}

\maketitle

\begin{abstract}
    This paper considers the problem of label recovery in random graphs and matrices. Motivated by transitive behavior in real-world networks (i.e., ``the friend of my friend is my friend''), a recent line of work considers spatially-embedded networks, which exhibit transitive behavior. In particular, the Geometric Hidden Community Model (GHCM), introduced by Gaudio, Guan, Niu, and Wei, models a network as a labeled Poisson point process where every pair of vertices is associated with a pairwise observation whose distribution depends on the labels and positions of the vertices. The GHCM is in turn a generalization of the Geometric SBM (proposed by Baccelli and Sankararaman). Gaudio et al. provided a threshold below which exact recovery is information-theoretically impossible. Above the threshold, they provided a linear-time algorithm that succeeds in exact recovery under a certain ``distinctness-of-distributions'' assumption, which they conjectured to be unnecessary. In this paper, we partially resolve the conjecture by showing that the threshold is indeed tight for the two-community GHCM. We provide a two-phase, linear-time algorithm that explores the spatial graph in a data-driven manner in Phase I to yield an almost exact labeling, which is refined to achieve exact recovery in Phase II. Our results extend achievability to geometric formulations of well-known inference problems, such as the planted dense subgraph problem and submatrix localization, in which the distinctness-of-distributions assumption does not hold.

\end{abstract}

\section{Introduction}
\label{sec:intro}

Many prominent inference problems on graphs and matrices in the statistics, machine learning, and information theory literature aim to recover hidden community structure, such as the Stochastic Block Model (SBM), group synchronization, and the spiked Wigner model. They share the same framework: graph vertices, or equivalently matrix indices, are assigned unobserved community labels, and we observe independent measurements between each pair of vertices. The distribution of these measurements is governed by the community labels of the corresponding vertices. For example, in the SBM, introduced by \cite{Holland1983} to analyze social networks, the measurements are Bernoulli random variables indicating the presence or absence of edges, where the Bernoulli parameter for a pair $(u,v)$ depends on labels of $u$ and $v$. In $\Z_2$ synchronization, the simplest form of group synchronization with two communities, the measurement between two vertices $u$ and $v$ is distributed as
an unit-variance Gaussian with mean $+\mu$ if they are from the same community and mean $-\mu$ if they are from opposite communities.

There is a surging interest in extending these inference problems to include latent geometry, in order to better capture empirical properties of real-world networks. In particular, social networks often display \emph{transitive} behavior, containing many triangles because a given pair of individuals is more likely to be friends if they have a friend in common \citep{rapoport1953spread}. In addition to transitivity considerations, spatial networks also model scenarios with limited information; observations may only be available between nearby vertices, or, by viewing the geometry as a feature embedding, between sufficiently similar data points.
Several spatial random graph models were introduced to enrich the SBM, such as those by \cite{Galhotra2018}, \cite{Sankararaman2018}, and \cite{Avrachenkov2021}.

Recent works have established sharp information-theoretic (IT) thresholds akin to those for the standard SBM. The works of \cite{Sankararaman2018} and \cite{Gaudio2023} together establish the IT threshold for exact recovery in the Geometric SBM. \cite{ghcm} generalized the model further to the Geometric Hidden Community Model (GHCM), which captures a wider variety of inference problems including $\Z_2$ synchronization under the geometric setting. They provided a threshold below which exact recovery is information-theoretically impossible and above which exact recovery is possible under a ``distinctness-of-distributions'' assumption: for every pair of vertices $u,v$ from distinct communities and $w \not \in \{u,v\}$, the distributions of measurements $(u, w)$ and $(v, w)$ must be distinct. This assumption is violated by several well-known inference problems such as submatrix localization and the planted dense subgraph problem. 
In this paper, we complete the picture of IT thresholds for the two-community GHCM by removing the distinctness assumption. We additionally provide a linear-time algorithm to achieve exact recovery.

\section{Model and Main Results}
\label{sec:main}

First, we formally define the GHCM and present existing results on exact recovery. 

\begin{definition}[Geometric Hidden Community Model]\label{def:ghcm}
Let $n\in\mathbb{N}$. Fix $\lambda > 0$, $d \in \mathbb{N}$, and $k\in\N$. Denote $Z\subset\Z$ with $|Z|=k$ as the set of community labels with corresponding prior probabilities $\pi\in\R^{k}$. For each $i,j \in Z$, let $P_{ij}$ be a probability distribution with mass/density function $p_{ij}$ on $\mathbb{R}$ such that $P_{ij} \stackrel{d}{=} P_{ji}$.
 
A graph $G$ is sampled from $\text{GHCM}(\lambda, n, \pi, P, d)$, with observations $\{Y_{uv}\}\subset\R$ over the undirected edges, according to the following steps:
\begin{enumerate}
    \item Generate the locations of vertices via a homogeneous Poisson point process with intensity $\lambda$ in the region $\cS_{d,n} := [-n^{1/d}/2, n^{1/d}/2 ]^d \subset \mathbb{R}^d$. Let $V\subset\cS_{d,n}$ denote the vertex set. \label{sample:step-1}
    \item Independently assign communities by $\pi$, such that the correct label of $u \in V$ is $\mathbb{P}(x^\star(u) = i) = \pi_i$ for $i\in Z$.
    \item Conditioned on locations and community labels, independently sample  pairwise observations $Y_{uv}$. For $u, v \in V$ and $u \neq v$, we have $Y_{uv} \sim P_{x^{\star}(u), x^{\star}(v)}$ if $\|u-v\|\le (\log n)^{1/d}$; otherwise $Y_{uv} = 0$. 
\end{enumerate}
Here $\lVert u-v\rVert$ denotes the toroidal metric to eliminate boundary effects:
$
\lVert u-v\rVert^2 = \sum_{i=1}^d \min\{|u_i-v_i|, n^{1/d}-|u_i-v_i|\}^2$.
\end{definition}

Given $\{Y_{uv}\}$ and the geometric locations of vertices, our goal is to estimate 
$x^\star$ up to some level of accuracy. We are interested in (1) almost exact recovery, where we recover all labels but a vanishing fraction of vertices, and (2) exact recovery, where we recover all labels with high probability. If symmetries are present in $P$ and $\pi$, then estimation is only correct up to a permutation, which is characterized by the following definition.

\begin{definition}[Permissible relabeling] \label{def:permissible} A permutation $\omega\colon Z\to Z$ is a \emph{permissible relabeling} if $\pi_i = \pi_{\omega(i)}$ for any $i\in Z$ and $P_{ij} = P_{\omega(i), \omega(j)}$ for any $i,j\in Z$. Let $\Omega_{\pi, P}$ be the set of permissible relabelings.
\end{definition}

Given an estimator $\widetilde{x}$, define the agreement of $\widetilde{x}$ and $x^\star$ as
\[
A(\widetilde{x}, x^\star) = \frac{1}{|V|} \max_{\omega\in\Omega_{\pi, P}} \sum_{u\in V}\mathds{1}\{\widetilde x(u) = \omega\circ x^\star(u)\},
\] 
and define
  \begin{itemize}
      \item \emph{Exact recovery:} $\lim\limits_{n \to \infty} \bP(A(\widetilde{x}, x^\star)=1) = 1$,
      \item \emph{Almost exact recovery:} $\lim\limits_{n \to \infty} \bP(A(\widetilde{x}, x^\star)\ge1-\epsilon) = 1$, for all $\epsilon>0$.
  \end{itemize}

\cite{ghcm} identified the condition below which exact recovery is information-theoretically impossible, using the Chernoff-Hellinger (CH) divergence measure on the edge observations. Let $p=(p_1, \hdots, p_k), q=(q_1, \hdots, q_k)$ be vectors of distributions associated with a vector of prior probabilities $\pi=(\pi_1, \hdots, \pi_k)$. The CH divergence between $p$ and $q$ is 
$D_+(p,q) = 1 - \inf_{t \in [0,1]} \sum_{i=1}^k \pi_i \sum_{x \in \mathcal{X}} p_i(x)^t q_i(x)^{1-t},$
if $p, q$ are discrete. In the continuous case, the inner summation is replaced by an integral.

\cite{ghcm} showed that any estimator fails to achieve exact recovery for $G\sim \text{GHCM}(\lambda, n, \pi, P, d)$ if $\lambda \nu_d \min_{i \neq j}  D_+(\theta_i, \theta_j) < 1$, where $\nu_d$ is the volume of the unit ball in $\R^d$ and $\theta_i = (p_{i1}, \cdots, p_{ik})$  associated with $\pi$. 
When $\lambda \nu_d \min_{i \neq j}  D_+(\theta_i, \theta_j) > 1$,
\cite{ghcm} provided a linear-time algorithm that achieves exact recovery under two assumptions (and hence establishes $\lambda \nu_d \min_{i \neq j}  D_+(\theta_i, \theta_j) = 1$ as the IT threshold under these assumptions):
\begin{enumerate}
    \item The log-likelihood ratio $\log(p_{ij}(\cdot) / p_{ab}(\cdot))$ is bounded.
    \item Distinctness: $P_{ir} \neq P_{is}$ for $i, r\neq s\in Z$.
\end{enumerate}
\cite{ghcm} proposed a two-phase algorithm for exact recovery: Phase I constructs a preliminary labeling which achieves almost exact recovery, and Phase II refines the preliminary labeling to achieve exact recovery. Phase I relies on a construction, the \textit{visibility graph}, made by partitioning the torus into small blocks $\{B_i\}$. The visibility graph is a meta-graph, in which the nodes correspond to $\{B_i\}$ and an edge is formed between two blocks, denoted as $B_i \sim B_j$, if $\lVert x - y \rVert \leq (\log n)^{1/d}$ for all $x\in B_i, y\in B_j$. Above the IT threshold, \cite{ghcm} showed the visibility graph is connected with high probability. Phase I exploits this property by constructing a minimal spanning tree on the visibility graph. It labels the vertices within the root node of the tree using the \textit{maximum a posteriori} (MAP) estimator. Then, using a propagation schedule on the tree (e.g., a depth-first search), Phase I labels each vertex of a block using its edge observations to the vertices in the preceding block. 
\cite{ghcm} showed this labeling yielded a sufficiently low error probability and achieved almost exact recovery. Phase II mimicked the $\textit{genie-aided estimator}$, which labels each vertex via MAP estimation using the correct labels of all visible vertices. Since the correct labels are not available, \cite{ghcm} replaced them with the Phase I labels and showed such refining achieves exact recovery. 

While the first assumption bounding the log-likelihood ratio can be relaxed in some cases by modifying the proof technique (e.g., for Gaussian $P_{ij}$'s, see \cite[Appendix G]{ghcm}), the second assumption of distinct distributions leaves the precise characterization of the IT threshold unclear in some prominent inference problems under the geometric setting. For example, consider the planted dense subgraph (PDS) problem, which aims to recover a subgraph $C^\star$ in a random graph where edges within 
$C^\star$ are formed with probability $p$, and edges between vertices outside or crossing $C^\star$ are formed with probability $q<p.$ The geometric extension of PDS can be formulated as a GHCM, with $k=2$, $Z=\{1, 2\}$, and $P$ such that $P_{11} = \text{Bern}(p)$ and $P_{12} = P_{21} = P_{22} = \text{Bern}(q)$, violating the distinctness assumption.
Since the tree search order of \cite{ghcm} is determined by the number of vertices in each block but not their labels, there may be blocks which do not contain any vertices from $C^{\star}$, and thus fail to label their children. Another model violating the distinctness assumption is submatrix localization (SL), which aims to recover the principal submatrix whose entries are drawn from a Gaussian with elevated mean in a matrix whose remaining entries are zero-mean Gaussians. 
The geometric SL is a GHCM, with $k=2$, $Z=\{1, 2\}$, and $P$ such that $P_{11} = \cN(\mu, 1)$ and $P_{12} = P_{21} = P_{22} = \cN(0, 1)$ for $\mu > 0.$

In this paper, we show that $\lambda \nu_d \min_{i \neq j}  D_+(\theta_i, \theta_j) = 1$ is indeed the IT threshold for exact recovery in the two-community GHCM, by showing that above the threshold, exact recovery is possible even without the distinctness assumption.
In the rest of the paper, we set $k=2$, $Z=\{1, 2\}$, and $P$ such that $P_{11} \neq P_{12}, P_{21} = P_{22}$. We let $C^\star = \{u\in V: x^\star(u)=1\}$. Our main result establishes the sharpness of the IT threshold.
\begin{theorem}
\label{thm:main_exact_recovery}
Let $G\sim \text{GHCM}(\lambda, n, \pi, P, d)$, $k=2, d\geq 2$, such that (i) $\log(p_{12}(y_{u_{ij}v})/p_{11}(y_{u_{ij}v}))<\eta$, or (ii) $P_{ij}$ are all Gaussians. Then, there exists an efficient algorithm that achieves exact recovery when $\lambda \nu_d \min_{i \neq j}  D_+(\theta_i, \theta_j) > 1$ and almost exact recovery when $\pi_1\lambda \nu_d > 1$. 
\end{theorem}

A key observation is that the refinement from almost exact recovery to exact recovery of \cite{ghcm} succeeds without the distinctness assumption. Therefore, our main contribution is an almost exact labeling procedure that yields exact recovery when combined with the refinement step for any $2$-community GHCM. The algorithm explores the visibility graph in a data-driven manner using the edge observations, in contrast with the previous approach which fixes the block visitation order \textit{a priori} of observing the edges.

\section{Exact Recovery Algorithm for the general 2-community GHCM}
\label{sec:algo}

The algorithm is a two-phase approach that yields an almost exact labeling $\widehat x$ in Phase I and an exact labeling $\widetilde x$ in Phase II (Algorithm \ref{alg:full}). The runtime is $O(n\log n)$, which is linear in the number of edges of $G$. Denote $\widehat x_S$ ($\widetilde x_S$) as the subset of $\widehat x$ ($\widetilde x)$ restricted to $S \subset V.$ For $u, v\in V$, denote $u \sim v$ if $\lVert u - v\rVert \leq (\log n)^{1/d}$, i.e., $u$ and $v$ are mutually visible and the edge observation $Y_{uv}$ is available.

\subsection{Phase I}
The algorithm first partitions $\cS_{d, n}$ into disjoint blocks of volume $\chi \log n$ (Line \ref{line:partition}), where $\chi >0$ is a sufficiently small constant precisely defined in Section \ref{sec:phaseI}. Denote $V(B)$ as the set of vertices in a block $B$. The algorithm considers blocks that contain sufficiently many vertices, formally defined below.
\begin{definition}
      Given $\delta>0$, a block $B \subset \cS_{d,n}$ is \emph{$\delta$-occupied} if $|V(B)| > \delta\log n$. Otherwise, $B$ is $\delta$-unoccupied.
\end{definition}
We say $B$ is occupied if $\delta$ is clear. For sufficiently small $\delta$, all blocks are occupied except a negligible fraction. Therefore, the algorithm can only consider the occupied blocks to achieve almost exact recovery. 

The algorithm begins labeling by selecting an arbitrary occupied block $B_1$ and conducting MAP estimation on a subset of the block's vertices $V_0$ (Line \ref{line:MAP}). We only label a subset to maintain linear runtime. By Theorem D.2 of \cite{ghcm}, MAP achieves exact recovery on $V_0$. Then, the algorithm propagates the labeling using Algorithm \ref{alg:propagation}, which takes a labeled set of vertices and labels another set of vertices by maximizing the likelihood of the edge observations between the two sets. The first propagation from $V_0$ to the remaining vertices of $B_1$ yields a completely labeled block (Line \ref{line:prop_V1}). 

The algorithm continues label propagation by labeling neighboring blocks in the visibility graph 
in a data-driven manner that depends on the occupancy of $\cstar$ nodes (Lines \ref{line:prop_begin}-\ref{line:prop_end}).
\begin{definition}
    \label{def:occupancy_cstar} A block is \emph{$\cstar$-occupied} if the block contains at least $\delta \log n$ vertices from $C^\star.$ The block is \emph{$\cstar$-identified} if it contains at least $\delta \log n /2$ vertices from $C^\star.$
\end{definition}
Moreover, we say that a block is ``active'' if it has been labeled, but not propagated from. Given an active block $B_i$, let $N(B_i)$ denote the set of unlabeled, $\cstar$-occupied blocks that are visible to $B_i.$ The algorithm propagates from $B_i$ to all blocks in $N(B_i)$. Afterwards, $B_i$ is considered ``explored,'' and the algorithm moves to another active block to continue propagation until there are no more active blocks.

The success of the propagation procedure hinges on the fact that all $C^{\star}$-occupied blocks are visited.  To understand why propagation should have this property, consider the \emph{vertex visibility graph} on $C^{\star}$, in which the nodes are vertices of the GHCM restricted to $C^{\star}$ and edges are drawn between every pair of vertices within distance $(\log n)^{1/d}$. When the model parameters are above the IT threshold, the result of \cite{Penrose1997} implies that the vertex visibility graph on $\cstar$ is connected with high probability. It turns out that a coarsening of the vertex visibility graph on $C^{\star}$, which we refer to as simply as the $C^{\star}$-\emph{visibility graph}, is also connected with high probability. Formally, a visibility graph is defined as follows.

  \begin{definition}[Visibility graph]
  Consider a Poisson point process $V\subset\cS_{d,n}$, the $(\chi\log n)$-block partition of $\cS_{d,n}$, $\{B_i\}_{i=1}^{n/(\chi\log n)}$, corresponding vertex sets $\{V_i\}_{i=1}^{n/(\chi\log n)}$, and a constant $\delta>0$. The \emph{$(\chi, \delta)$-visibility graph} is denoted by $H = (V^{\dagger}, E^{\dagger})$, where the vertex set $V^{\dagger} = \{i \in [n/(\chi\log n)] : |V_i| \geq \delta \log n\}$ consists of all $\delta$-occupied blocks and the edge set is given by $E^{\dagger}=\{\{i,j\}\colon i,j \in V^{\dagger}, B_i \sim B_j\}$.
  \end{definition}
We drop $(\chi, \delta)$ when these constants are evident and further define the \emph{$\cstar$-visiblity graph} as the visibility graph formed by vertices restricted to $\cstar$.
\begin{proposition}
    \label{prop:cstar}
    Let $\chi, \delta$ satisfy \eqref{eq:chi-formula} and \eqref{eq:delta-formula}. When $\lambda \nu_d \min_{i \neq j}  D_+(\theta_i, \theta_j) > 1$, the $\cstar$-visibility graph is connected with high probability.
\end{proposition}

Let $H$ be the $C^{\star}$-visibility graph. Although $H$ is unknown, we can discover its vertices through a graph exploration, which explores blocks that are estimated to contain at least $\delta \log n/2$ vertices in $C^{\star}$. Since the estimation makes few errors with high probability, we show that the exploration process finds all vertices in $H$ with high probability. After propagation, the algorithm labels all remaining vertices as not in $\cstar$ (Line \ref{line:remain}). The labeling is default since the exploration process misses only the blocks outside of $H$, and hence only misses blocks which contain at most $\delta \log n$ vertices from $C^{\star}$.

Under the same conditions as Proposition \ref{prop:cstar}, Phase I achieves almost exact recovery.
\begin{theorem}
    When $ \lambda \nu_d \min_{i \neq j}  D_+(\theta_i, \theta_j) > 1$, Phase I of Algorithm \ref{alg:propagation} achieves almost exact recovery. \label{thm:almost-exact}
\end{theorem}
We prove Proposition \ref{prop:cstar} and Theorem \ref{thm:almost-exact} in Section \ref{sec:phaseI}.

\subsection{Phase II}
Phase II is a refinement phase that improves the labeling of Phase I by mimicking the genie-aided estimator. Given the correct labeling of all vertices visible to a vertex $u$, the genie-aided estimator is Bayes optimal because it maximizes the likelihood of $Y_{uv}$:
\[\widetilde{x}_{\text{genie}}(u) = \argmax_{i\in \{1, 2\}}\sum_{v\in V\setminus\{u\}, v\sim u} \log p_{i, x^\star(v)} (y_{uv}). \]
However, as $x^\star$ is not available, the algorithm replaces $x^\star$ with $\widehat x$ (Line \ref{line:refine}). The refinement procedure is identical to that of \cite{ghcm}[Theorem E.3], which assumes only that the model parameters are above the IT threshold and the preliminary labeling is such that all blocks have at most $\epsilon \log n$ errors for $\epsilon > 0$ sufficiently small. Crucially, the existing refinement procedure does not require the distinctness assumption. 
Therefore, given that our Phase I produces an almost exact labeling, Algorithm \ref{alg:full} achieves exact recovery.

 \begin{algorithm}
    \caption{\texttt{Propagate}} \label{alg:propagation}
    \begin{algorithmic}[1]
    \Require{ Graph $G=(V, E)$, mutually visible vertex sets $T, T' \subset V$, $T \cap T' = \emptyset$, and $\widehat{x}_T\colon T\to \{1, 2\}$.}
    \Ensure{An estimated labeling $\widehat x_{T'} \colon T'\to \{1, 2\}$.} 
    \For{$u \in T'$}
        \[\widehat x_{T'}(u) = \argmax_{r\in \{1, 2\}}  \sum_{v\in T, v\sim u, \widehat x_T(v)=1} \log p_{1r}(y_{uv}).\]
    \EndFor
    \end{algorithmic}
\end{algorithm}

\begin{algorithm}
      \caption{Exact recovery for the 2-community GHCM} \label{alg:full}
      \begin{algorithmic}[1]
      \Require $G \sim \text{GHCM}(\lambda, n, \pi, P, d)$.
      \Ensure{ An estimated community labeling $\widetilde{x}: V \to Z$.}
      \vspace{5pt}
      \State{{\bf Phase I:}} 
      \State Take $\chi,\delta>0$ satisfying the conditions \eqref{eq:chi-formula} and \eqref{eq:delta-formula}. 
      \State Partition $\cS_{d,n}$ into blocks of volume $\chi\log n$. \label{line:partition}
      \State Select an initial block $B_{1}$.  \label{line:initial_select}
      \State Set $\varepsilon_0 \leq \min\{1/(2\log 2), \delta\}$. Sample $V_{0}\subset V_{1}$ such that $|V_{0}| = \varepsilon_0 \log n$. Set $V_{1}^\prime \leftarrow V_{1} \setminus V_{0}$. 
  \State{Label $V_{0}$ using Maximum a Posteriori estimation, i.e., set} \label{line:MAP} 
  \[\widehat x_{V_{0}} = \argmax_{x\colon V_{0}\to \{1, 2\}} \mathbb{P}(x^\star=x | G).\]
  \State{If $V_{1}^\prime \neq \emptyset$, apply \texttt{Propagate} (Algorithm \ref{alg:propagation}) on input $G, V_{0}, V_{1}^\prime$ to determine the labeling $\widehat{x}$ on $V_{1}^\prime$.} \label{line:prop_V1}
  \State Set the indices of \emph{active} and \emph{explored} blocks to be $\cA=\{1\}$ and $\cE = \emptyset$, respectively.  
  \While{$|\mathcal{A}| > 0$} \label{line:prop_begin} 
  \State Select an arbitrary $i \in \mathcal{A}$.
  \For{$j \in N(B_i) \setminus \cE$}
   \State Apply \texttt{Propagate} (Algorithm \ref{alg:propagation}) on input $G, V_{i}, V_{j}$ to determine the labeling $\widehat{x}$ on $V_{j}$.
   \If{$\widehat{x}$ labels at least $\frac{\delta \log n}{2}$ vertices as $\cstar$ in $V_j$}
   \State Append $j$ to $\cA$.
   \EndIf
  \EndFor
  \State Remove $i$ from $\cA$. Append $i$ to $\cE$.
  \EndWhile \label{line:prop_end}
     
  \For{$u \in V \setminus \left(\cup_{i \in \cE} V_i\right)$} \label{line:rem_start}
  \State Set $\widehat{x}(u) = 2$. \label{line:remain}
  \EndFor 
  \vspace{5pt}
  \State{{\bf Phase II:}}
  \For{$u\in V$}  
  \State \[\widetilde{x}(u) = \argmax_{i\in \{1, 2\}}\sum_{v\in V\setminus\{u\}, v\sim u} \log p_{i,\widehat x(v)} (y_{uv})\] \label{line:refine}
  \EndFor
  \end{algorithmic}
  \end{algorithm}

\section{Proof of Almost exact recovery}
\label{sec:phaseI}

\subsection{Proof of Proposition \ref{prop:cstar}}

\begin{proof}
We first show $\pi_1 \lambda \nu_d>1.$ We prove the inequality for the continuous case, as the discrete case is proven similarly. Since $P_{12}=P_{22}$, the CH divergence simplifies to 
    \begin{align*}
        D_+(\theta_1, \theta_2) &= \pi_1 \left(1 - \inf_{t \in [0,1]} \left\{  \int_{x \in \mathcal{X}} p_{11}(x)^t p_{12}(x)^{1-t} dx \right\} \right).
    \end{align*}
    Since the infimum is bounded above by 1, the condition $\lambda \nu_d \min_{i \neq j}  D_+(\theta_i, \theta_j) > 1$ immediately implies $\pi_1 \lambda \nu_d > 1.$

For a GHCM with rate $\lambda'$, \cite[Appendix C]{ghcm} defined the following constants 
\begin{align}
  &\nu_d \big(1 - 3\sqrt{d}\chi^{1/d}/2  \big)^d \ge \Big(\nu_d + \frac{1}{\lambda'} \Big)/2 \nonumber \\
  & 0<\chi < \Big( \nu_d- \frac{1}{\lambda'} \Big)/2, \label{eq:chi-formula} \\
  &R_d =1- \sqrt{d}\chi^{1/d}/2, \quad 0<\delta<  \frac{\widetilde{\delta}\chi}{\nu_dR_d} \label{eq:delta-formula}
  \end{align}
and showed the corresponding $(\chi, \delta)$-visibility graph is connected with probability $1-o(1)$ when $\lambda' \nu_d >1$, which holds whenever above the IT threshold. Since the $\cstar$ vertices are a Poisson process with intensity $\pi_1 \lambda$, replacing $\lambda'$ by $\pi_1 \lambda$ ensures that the resulting $\cstar$-visibility graph is connected.

\end{proof}

\subsection{Propagation}
Denote $N$ as the number of occupied blocks. Let $Z \in \N^{2\times N}$ be the count vector of the visibility graph, indicating the number of vertices of communities 1 and 2 in each block. While $Z$ is unobserved, it is convenient to condition on its realization. We introduce a set $S \subseteq \N^{2\times N}$ of ``nice'' count vectors $z$ which satisfy the following conditions:
\begin{enumerate}
    \item[(a)] the initial block has at least $\delta \log n$ vertices from $\cstar$.
    \item[(b)] $z$ induces connected $H$.
\end{enumerate}
Above the IT threshold, applying Lemma C.7 of \cite{ghcm} implies that the initial block formed by Line \ref{line:initial_select} is $\cstar$-occupied with high probability, satisfying Condition (a). Additionally, Proposition \ref{prop:cstar} ensures Condition (b) is satisfied. It follows that $\mathbb{P}(Z \in S) = 1-o(1)$.

Suppose that $Z \in S$, and all blocks which are included in the tree exploration are labeled perfectly. Since the tree exploration explores blocks with at least $\frac{\delta}{2} \log n$ vertices labeled as $C^{\star}$, Conditions (a) and (b) imply that all $C^{\star}$-occupied blocks will be visited by the tree exploration. In reality, the blocks are not perfectly labeled. However, if the algorithm makes at most $M$ mistakes on occupied blocks such that $\delta \log n - M \geq \frac{\delta}{2}\log n$, then we are guaranteed to include all occupied blocks in the tree. More formally, let $F_i$ be the event that the algorithm finishes building the tree before the $i^{\text{th}}$ visited block and let $A_i$ be the event that the $i^{\text{th}}$ visited block is ``successfully'' labeled, meaning that either (i) $F_i$ holds, and at most $\delta \log n$ vertices are mislabeled; or (ii) $F_i^c$ holds, and at most $M$ vertices are mislabeled. Let $A_0$ be the event that $V_0$ 
is labeled perfectly. The event $E_i \triangleq \left(\bigcap_{j=0}^i A_i \right) \cap \{Z \in S\} \cap F_i$ implies that the tree comprising the first $i$ visited blocks contains all blocks which are $C^{\star}$-occupied.

We desire a lower bound on $\bP(A_i | A_0, ... , A_{i-1}, Z = z)$. Observe $\bP(A_i | A_0, ... , A_{i-1}, Z = z, F_i) = 1$ because the $i^{\text{th}}$ block, which is $\cstar$-unoccupied, is labeled as being fully outside $C^{\star}$, incurring at most $\delta \log n$ mistakes. Then, by the law of total probability, we have that 
\begin{align}
    \bP(A_i | A_0, ... , A_{i-1}, Z = z) &\geq P(A_i | A_0, ... , A_{i-1}, Z = z, F_i^c), \nonumber
\end{align}
where the right-hand side is the success probability of labeling the $i^{\text{th}}$ block according to its parent in the tree. 

Next, we characterize the misclassification rate when labeling a single vertex in the propagation scheme. Denote $u_{ij}$ as the $j$-th vertex in the $i^{\text{th}}$ encountered block of the algorithm. 

\begin{lemma}
    \label{lem:misclass_rate} Suppose $ \lambda \nu_d \min_{i \neq j}  D_+(\theta_i, \theta_j) > 1$ and one of the following holds: (i) $
    \log(p_{12}(y)/p_{11}(y))<\eta$ for $ y \in \R$, or (ii) $P_{ij}$ are all Gaussians. Given $F_i^c, Z$ and the successful labeling history $A_1, \cdots, A_{i-1}$, there exist constants $a, b>0$ such that
    \[ \bP(\widehat{x}(u_{ij}) \neq x^\star(u_{ij}) | A_0, ... , A_{i-1}, Z = z, F_i^c)  \leq a  n^{-b}.  \]
\end{lemma}

\begin{proof}
Denote $\bP_{i, z}(\cdot) = \bP(\cdot \mid A_0, ... , A_{i-1}, Z = z, F_i^c)$ and $\widehat C_i$ as the set of vertices labeled as $\cstar$ in the $i^{\text{th}}$ encountered block. Under assumption (i), a straightforward adaptation of \cite[Lemma D.7]{ghcm} gives that 
\begin{align}
\bP_{i, z}(\widehat{x}(u_{ij}) \neq x^\star(u_{ij})) 
&\leq \phi^{\frac{\delta}{2} \log n - M} \cdot e^{\eta M} \nonumber \\
&=\left( \frac{e^\eta}{\phi}\right)^M n^{-\delta \log(1 / \phi) / 2}. \label{eq:off_lrt}
\end{align}
where $\phi = \sup_{t\in (0, 1)} \int p_{11}(y)^t p_{12}(y)^{1-t} dy$. Note that $\phi < 1$ since $p_{11} \neq p_{12}$.

Setting $a=(e^\eta / \phi)^M, b=\delta \log(1 / \phi) / 2$ yields the desired result. Under assumption (ii), observe that $\log(p_{12}(y_{u_{ij}v})/p_{11}(y_{u_{ij}v}))$ is a Gaussian and hence the expectation is the moment-generating function of a Gaussian that can be bounded following similar derivations in \cite[Appendix G]{ghcm} to yield a similar bound as \eqref{eq:off_lrt}.
\end{proof}

Next, we lower bound $\bP(A_i | A_0, ... , A_{i-1}, Z = z, F_i^c)$. By \cite[Lemma D.1]{ghcm}, there exists $\Delta$ such that $|V_i| < \Delta \log n$ for all $i \in V^\dagger$ with high probability. For a fixed $i$ and conditioned on $|\widehat{C}_{p(i)} \cap C^{\star}|$ and $|\widehat{C}_{p(i)} \setminus C^{\star}|$, the events that $u_{ij}$ is misclassified by the algorithm are mutually independent given $A_1, \cdots, A_{i-1}, F_i^c$ and $Z$. By Lemma \ref{lem:misclass_rate}, the total number of errors on $V_i$ is stochastically dominated by $X \sim \text{Bin}(\Delta \log n, an^{-b})$ for some $a, b.$ 
Denoting $\mu=\E[X]$, a Chernoff bound yields
\begin{align*}
    \bP_{i, z}(A_i^c ) &= \bP_{i, z}(\lvert u_{ij}\in V_i: \widehat x(u_{ij})\neq x^\star(u_{ij}) \rvert > M) \\
    &\leq \bP(X > M) \\
    &= \bP(X - \mu > (M/ \mu - 1)\mu) \\
    &\leq e^{M - \mu} (\mu / M)^M \leq e^{M } (\mu / M)^M\\
    &\leq (e\Delta a / M)^M (\log n)^M n^{-bM}.
\end{align*}
Set $M=5/(4b)$ and define $\gamma = (e\Delta a / M)^M$. For large enough $n$, we have $(\log n)^M < n^{1/8}$, yielding
\begin{align}
    &\bP(A_i | A_0, ... , A_{i-1}, Z = z) \geq \bP_{i, z}(A_i) \geq 1- \gamma n^{-9/8}. \nonumber
\end{align}

We will now prove almost exact recovery of Phase I. By \cite[Theorem D.2]{ghcm}, $A_0$ holds with high probability. Then,
\begin{align}
    &\bP(\bigcap_{i\in [N]} A_i \mid Z=z) \nonumber \\
    &\geq \bP(A_0 \mid Z=z) \cdot \prod_{i=1}^{N}\bP(A_i \mid A_0, \dots, A_{i-1}, Z=z) \nonumber \\
    &\geq (1-o(1)) \Big( 1- \gamma n^{-9/8} \Big)^{\frac{n}{\chi \log n}} \nonumber \\
    &\geq (1-o(1)) \Big( 1- \frac{\gamma n^{-1/8}}{\chi \log n} \Big), \nonumber
\end{align}
where the last inequality is due to Bernoulli's inequality. The bound is uniform over all $z\in S$, with $\{Z\in S\}$ occurring with high probability. Since $M < \delta \log n$ for $n$ large enough and there are $n/\chi \log n$ blocks in total, $\widehat x$ makes fewer than $\delta \log n \cdot n/\chi \log n = \delta n/ \chi$ mistakes in total, yielding almost exact recovery.

\section{Conclusion}
\label{sec:conclusion}

In this paper, we present a linear-time algorithm for exact recovery in the GHCM with two communities. We partially resolve the conjecture of \cite{ghcm} by showing that efficient exact recovery is achievable above the IT threshold of $\lambda \nu_d \min_{i \neq j}  D_+(\theta_i, \theta_j) = 1$ in two communities. Our algorithm extends achievability to well-known inference problems such as the planted dense subgraph problem and submatrix localization. 

A natural question is whether our approach can be extended to the $k$-community GHCM with $k \geq 3$. One can show that for two communities $i\neq j$, the subgraph of the GHCM consisting of communities $\{l: P_{il} \neq P_{jl}\}$ that yield distinct distributions to $i$ and $j$ form a connected visibility graph whenever above the IT threshold. Such a visibility graph would be useful to distinguish Communities $i$ and $j$, but how can we create it? 
When $k=2$, the community $C^{\star}$ is enough to distinguish the communities, while when $k \geq 3$, there can be logical interdependencies in how communities distinguish each other. Therefore, new algorithmic ideas are needed to handle the $k \geq 3$ case.

\paragraph*{Acknowledgements.}
J.G. and C.K.G. were partially supported by NSF CCF 2154100, and C.K.G. was partially supported through NSF HDR TRIPODS (IDEAL). We thank Xiaochun Niu and Ermin Wei for helpful discussions.

\bibliographystyle{plain} 
\bibliography{refs} 

\begin{thebibliography}{1}

\bibitem{Avrachenkov2021}
Konstantin Avrachenkov, Andrei Bobu, and Maximilien Dreveton.
\newblock Higher-order spectral clustering for geometric graphs.
\newblock {\em Journal of Fourier Analysis and Applications}, 27(2):22, 2021.

\bibitem{Galhotra2018}
Sainyam Galhotra, Arya Mazumdar, Soumyabrata Pal, and Barna Saha.
\newblock The geometric block model.
\newblock In {\em Proceedings of the AAAI Conference on Artificial Intelligence}, volume~32, 2018.

\bibitem{ghcm}
Julia Gaudio, Charlie Guan, Xiaochun Niu, and Ermin Wei.
\newblock Exact label recovery in euclidean random graphs, 2024.

\bibitem{Gaudio2023}
Julia Gaudio, Xiaochun Niu, and Ermin Wei.
\newblock Exact community recovery in the geometric sbm.
\newblock {\em Symposium on Discrete Algorithms (SODA)}, 2024.

\bibitem{Holland1983}
Paul~W Holland, Kathryn~Blackmond Laskey, and Samuel Leinhardt.
\newblock Stochastic blockmodels: First steps.
\newblock {\em Social Networks}, 5(2):109--137, 1983.

\bibitem{Penrose1997}
Mathew~D. Penrose.
\newblock {The longest edge of the random minimal spanning tree}.
\newblock {\em The Annals of Applied Probability}, 7(2):340 -- 361, 1997.

\bibitem{rapoport1953spread}
Anatol Rapoport.
\newblock Spread of information through a population with socio-structural bias: {I. Assumption} of transitivity.
\newblock {\em The Bulletin of Mathematical Biophysics}, 15:523--533, 1953.

\bibitem{Sankararaman2018}
Abishek Sankararaman and Fran{\c{c}}ois Baccelli.
\newblock Community detection on {Euclidean} random graphs.
\newblock In {\em Proceedings of the Twenty-Ninth Annual ACM-SIAM Symposium on Discrete Algorithms}, pages 2181--2200. SIAM, 2018.

\end{thebibliography}

\end{document}